\documentclass[conference]{IEEEtran}
\IEEEoverridecommandlockouts
\usepackage{cite}
\usepackage{amsmath,amssymb,amsfonts}
\usepackage{subcaption}
\usepackage{graphicx}
\usepackage{textcomp}
\usepackage{xcolor}
\usepackage{complexity}
\usepackage{multirow}
\usepackage{amsthm}
\usepackage{accents}
\usepackage{mathtools}
\usepackage{mathrsfs}
\usepackage{bm}
\usepackage{csquotes}
\usepackage[ruled, linesnumbered, noend]{algorithm2e}

\newtheorem{theorem}{Theorem}[section]
\newtheorem{corollary}[theorem]{Corollary}
\newtheorem{lemma}[theorem]{Lemma}

\newtheorem{proposition}[theorem]{Proposition}
\newtheorem{definition}[theorem]{Definition}
\newtheorem{example}[theorem]{Example}

\usepackage{tikz}
\usetikzlibrary{positioning}
\usetikzlibrary{calc}
\usetikzlibrary{patterns}
\usetikzlibrary{arrows.meta}
\usetikzlibrary{automata, positioning, arrows.meta}
\usetikzlibrary{3d}
\usetikzlibrary{decorations.pathmorphing}
\pgfdeclarelayer{bg}    
\pgfsetlayers{bg,main}  

\usepackage{tabularx}
\usepackage{array}
\usepackage{ragged2e}

\makeatletter
\providecommand{\bigsqcap}{%
  \mathop{%
    \mathpalette\@updown\bigsqcup
  }%
}
\newcommand*{\@updown}[2]{%
  \rotatebox[origin=c]{180}{$\m@th#1#2$}%
}
\makeatother

\newcolumntype{L}{>{\RaggedRight\arraybackslash}X}
\newcolumntype{C}{>{\centering\arraybackslash}X}
\newcolumntype{N}{>{\RaggedRight\arraybackslash}m{4cm}} 
\newcolumntype{T}{>{\RaggedRight\arraybackslash}m{3.7cm}} 
\newcolumntype{W}{>{\raggedright\arraybackslash}m{5.7cm}} 
\newcolumntype{F}{>{\centering\arraybackslash}m{3.25cm}} 

\newcommand{\hcell}[1]{\parbox[c][1.5\baselineskip][c]{\linewidth}{#1}}

\newcommand{\CC}{{\mathscr{L}}}
\newcommand{\CCP}{{\mathscr{P}}}

\begin{document}

\title{Constrained Cuts, Flows, and Lattice-Linearity\thanks{Partially supported by NSF CNS-1812349, and the Cullen Trust for
Higher Education Endowed Professorship.}}

\author{\IEEEauthorblockN{Robert Streit}
\IEEEauthorblockA{\textit{Dept. of Electrical and Computer Engineering} \\
\textit{The University of Texas at Austin}\\
Austin, TX 78712, USA \\
rpstreit@utexas.edu}
\and
\IEEEauthorblockN{Vijay K. Garg}
\IEEEauthorblockA{\textit{Dept. of Electrical and Computer Engineering} \\
\textit{The University of Texas at Austin}\\
Austin, TX 78712, USA \\
garg@utexas.edu}
}

\maketitle

\begin{abstract}
In a capacitated directed graph, it is known that the set of all min-cuts forms a distributive lattice \cite{picard1980structure,queyranne1998minimizing}. 
Here, we describe this lattice as a regular predicate whose forbidden elements can be advanced in constant parallel time after precomputing a max-flow, so as to obtain parallel algorithms for min-cut problems with additional constraints encoded by lattice-linear predicates \cite{DBLP:conf/spaa/Garg20}.
Some nice algorithmic applications follow.
First, we use these methods to compute the irreducibles of the sublattice of min-cuts satisfying a regular predicate.
By Birkhoff's theorem \cite{birkhoff1937rings} this gives a succinct representation of such cuts, and so we also obtain a general algorithm for enumerating this sublattice.
Finally, though we prove computing min-cuts satisfying additional constraints is \NP-hard in general, we use poset slicing \cite{garg2001slicing,Garg06-tcs} for exact algorithms (with constraints not necessarily encoded by lattice-linear predicates) with better complexity than exhaustive search.
We also introduce \emph{$k$-transition predicates} and \emph{strong advancement} for improved complexity analyses of lattice-linear predicate algorithms in parallel settings,
which is of independent interest.
\end{abstract}


\section{Introduction}
Network problems are fundamental in combinatorial optimization, with applications such as transportation \cite{schrijver2002history}, communication reliability \cite{yazdi2010max}, and image processing \cite{jensen2022review}. Famously, the Max-Flow Min-Cut Theorem \cite{ford1956maximal} establishes that the maximum flow (henceforth max-flow) value from a source \(s\) to a sink \(t\) in a capacitated graph equals the minimum capacity of a cut separating them (henceforth \emph{min-cut}). 
So, the prevailing method for generating min-cuts is running efficient max-flow algorithms (like \cite{edmonds1972theoretical,dinitz2006dinitz,goldberg1988new,orlin2013max,chen2025maximum})
and operating on the resulting residual graph. 

\begin{figure}[t]
\caption{An $(s,t)$-flow network with a unique min-cut and $2^n$ possible integral max-flows.
The min-cut is that given by partitioning around the edges intersecting with the dashed lined.
For any $i$, a max-flow  saturates exactly one of the paths through $a_i$ or $b_i$.
Hence, there are $2^n$ integral flows.}\label{fig:exp-maxflow}
\centering
\scalebox{.75}{
\begin{tikzpicture}

\node[draw, circle] (s) at (0,0) {$s$};
\node at (0,4.5) {$~$};

\node[draw, circle] (a1) at (3.25,3.75) {$a_1$};
\node[draw, circle] (b1) at (3.25,2.75) {$b_1$};
\node[draw, circle] (a2) at (3.25,1.5) {$a_2$};
\node[draw, circle] (b2) at (3.25,.5) {$b_2$};
\node at (3.25,-.4) {\Large$\vdots$};
\node[draw, circle] (an) at (3.25,-1.5) {$a_n$};
\node[draw, circle] (bn) at (3.25,-2.5) {$b_n$};

\node[draw, circle] (c1) at (6.5,2.75) {$c_1$};
\node[draw, circle] (c2) at (6.5,.5) {$c_2$};
\node at (6.5,-.4) {\Large$\vdots$};
\node[draw, circle] (cn) at (6.5,-1.5) {$c_n$};

\node[draw, circle] (t) at (9.75,0) {$t$};

\draw[->] (s) to node[above left] {$1$} (a1);
\draw[->] (s) to node[above right,xshift=2.5pt,yshift=7.5pt] {$1$} (b1);
\draw[->] (s) to node[above right,xshift=2.5pt,yshift=5pt] {$1$} (a2);
\draw[->] (s) to node[above right,xshift=2.5pt] {$1$} (b2);
\draw[->] (s) to node[below right,xshift=2.5pt,yshift=-2.5pt] {$1$} (an);
\draw[->] (s) to node[below left] {$1$} (bn);

\draw[->] (a1) to node[above] {$1$} (c1);
\draw[->] (b1) to node[below] {$1$} (c1);
\draw[->] (a2) to node[above] {$1$} (c2);
\draw[->] (b2) to node[below] {$1$} (c2);
\draw[->] (an) to node[above] {$1$} (cn);
\draw[->] (bn) to node[below] {$1$} (cn);

\draw[->] (c1) to node[above right] {$1$} (t);
\draw[->] (c2) to node[above] {$1$} (t);
\draw[->] (cn) to node[below right] {$1$} (t);

\draw[dashed, gray] (7.5,3.75) -- (7.5,-2.5);

\end{tikzpicture}}
\end{figure}

However, if one asks for a specific min-cut satisfying additional constraints for downstream applications, then generating the right min-cut is less straightforward.
Complicating the matter is the sizes of the solution spaces, and a duality lacking clear association between the identities of flows and cuts.
For example, there is an $(s,t)$-flow network with a single max-flow and an exponential number of min-cuts.
One places $n$ vertices, each with unit capacity edges between itself and $s$ and itself and $t$.
Then, the only max-flow will saturate every edge, while there are $2^n$ min-cuts corresponding to choosing either the edge to $s$ or edge to $t$ for each vertex.
See Figure~\ref{fig:hardness}.
Conversely, there is a network with a single min-cut and exponential number of max-flows,
which is best explained by examining Figure~\ref{fig:exp-maxflow}.
So, an efficient algorithm for min-cuts under additional constraints must use additional structure.

One such structure is the distributive lattice of min-cuts \cite{picard1980structure,queyranne1998minimizing}.
Recently, \cite{DBLP:conf/spaa/Garg20} showed that many problems on distributive lattices can be solved in a simple and unifying way using \emph{lattice-linear predicate detection}.
The idea is to model the problem as satisfaction of a \emph{lattice-linear} predicate, which informally is such that the satisfying preimage forms a semilattice under some meet operation.
Lattice-linear predicate detection can be solved by a universal procedure admitting simple parallel algorithmic implementations \cite{DBLP:conf/spaa/Garg20}.
Due to its generality, this idea has been applied to a diverse selection of problems detailed in Section~\ref{sec:related-work}.
And, although identifying min-cuts satisfying additional constraints is \NP-hard in general (see Section~\ref{sec:constraint-satisfying-mincuts}), by modeling min-cut computations through a lattice-linear predicate we provide an efficient method, Algorithm~\ref{alg:constrained-cuts}, for min-cuts satisfying additional constraints encoded by lattice-linear predicates.
We give some examples of lattice-linear predicates pertaining to min-cut problems in Section~\ref{subsec:prelim-mincut}.
In particular, we propose \emph{implicational cuts} that can be used to preserve some graph structure of interest in one side of the cut, and describe \emph{uniformly directed cuts}, examined in prior literature in the context of network reliability problems \cite{provan1989exact,provan1996paradigm} and efficient estimators of random variables on networks \cite{sigal1980stochastic,avramidis1996integrated}, as lattice-linear predicates as well.

\begin{table*}[t]
    \centering
    \caption{Our algorithmic contributions.  $T_B$ and $W_B$ are the (parallel) time and work complexities of computing forbidden vertices of the lattice-linear predicate over the cuts on a graph with $n$ vertices and $m$ edges, while $T_\text{MF}$ and $W_\text{MF}$ are the same for max-flows.
The exception to this is Algorithm~\ref{alg:slicing}, where $T_B$ and $W_B$ are the complexities of evaluating $B$.}
\label{tab:algorithm-complexities}
\begin{tabularx}{\textwidth}{|N|T W L|}
\hline
\textbf{Algorithm} & \textbf{Time} & \textbf{Work} & \textbf{Notes} \\
\hline\hline

\multirow{2}{=}{Min-cuts satisfying lattice-linear $B$ (Alg.~\ref{alg:constrained-cuts})}
&
\hcell{$O(T_\text{MF} + nT_B)$}
&
\hcell{$O(W_\text{MF} + n(m + W_B))$}
&
\hcell{General lattice-linear}
\\ \cline{2-4}
&
$O(T_\text{MF} + kT_B)$
&
$O(W_\text{MF} + k(m + W_B))$
&
$k$-transition with strong advancement
\\ \hline

\multirow{2}{=}{Irreducibles of min-cuts satisfying regular $B$ (Alg.~\ref{alg:irreducibles})}
&
\hcell{$O(T_\text{MF} + n T_B)$}
&
\hcell{$O(W_\text{MF} + n^2(m + n + W_{B}))$}
&
\hcell{General regular}
\\ \cline{2-4}
&
$O(T_\text{MF} + kT_B)$
&
$O(W_\text{MF} + kn(m + n + W_B))$
&
$k$-transition with strong advancement
\\ \hline

Enumerating min-cuts satisfying regular $B$ (Alg.~\ref{alg:enumeration})
&
$O(\log n)$ delay
&
$O(n^2)$ delay
&
Complexity per listed cut, precomputation omitted.
\\ \hline

Min-cut satisfying $B\wedge B_\text{reg}$, $B_\text{reg}$ regular and $B$ general (Alg.~\ref{alg:slicing})
&
$O(T_\text{MF}+nT_\text{reg} + (R\log n)\,T_B)$
&
$O\bigl(W_{\mathrm{MF}} 
      + n^2\bigl(m+n+W_{\mathrm{reg}}\bigr) 
      + (Rn^2)\,W_B\bigr)$
&
$R$ is number of min-cuts satisfying $B_\text{reg}$
\\ \hline

\end{tabularx}
\end{table*}

With that said, the appeal of lattice-linear predicate detection lies in its flexibility.
Because of this, we are able to derive many algorithmic applications in Section~\ref{sec:applic}.
Specifically, we begin by examining methods for computing the irreducibles of the sublattice of min-cuts satisfying a \emph{regular} predicate, that is a predicate whose satisfying preimage is closed under meet and join, in Section~\ref{subsec:irreducibles}.
This has downstream application towards constraint based extensions of recent works proposing efficient algorithms for \emph{diverse} solutions to min-cuts and other combinatorial optimization problems on distributive lattices which operate on the irreducibles \cite{de2023finding, de2025finding, iwamasa2025generalframeworkfindingdiverse}.
Furthermore, in Section~\ref{subsec:k-trans} we introduce \emph{$k$-transition predicates}, which allow us to get improved instance dependent guarantees for the parallel time complexities of lattice-linear predicate detection algorithms based on the structure of the given predicate.
We use this notion to tighten the analysis of our algorithms, and show that one can identify irreducible min-cuts (in absence of any other constraints) in parallel time equivalent to computing a single max-flow.
Then, using the succinct representation provided by the irreducibles, we propose an enumeration algorithm in Section~\ref{subsec:enumeration} based off traversing the ideals of the subposet of irreducible elements of min-cuts satisfying any regular predicate.
The generality of this method allows one to apply our algorithm to list the (unconstrained) min-cuts, the uniformly directed min-cuts, and any other constraint satisfying min-cuts defined by regular predicates without the need to modify the algorithm or rederive new results.
Finally, in Section~\ref{subsec:general} we use our enumeration method to obtain a poset-slicing \cite{garg2001slicing,Garg06-tcs} algorithm for identifying min-cuts satisfying general (possibly non-lattice-linear) predicates, with complexity better than simple exhaustive search.
Our algorithmic contributions are summarized in Table~\ref{tab:algorithm-complexities}, with their parallel time and work complexities explicitly listed.

\subsection{Organization}
We first survey the prior art in Section~\ref{sec:related-work}, and then review the preliminaries in Section~\ref{sec:preliminaries}.
Importantly, in Section~\ref{subsec:prelim-mincut} we present our method for modeling constraint satisfying minimum cuts on distributive lattices, and identify meaningful examples of constraints encoded by lattice-linear predicates.
Following this we present our technical contributions:
\begin{itemize}
    \item First, Section~\ref{sec:mincutLattice} proves min-cuts can be described as a regular predicate in Theorem~\ref{thm:min-cuts-ll}, whose \emph{forbidden} vertices can be computed via a single max-flow computation,
    \item Then, we define the \emph{constraint satisfying min-cuts} problem, show \NP-hardness, and propose Algorithm~\ref{alg:constrained-cuts} for efficient solution in the presence of constraints encoded by a lattice-linear predicate in Sections~\ref{sec:constraint-satisfying-mincuts} and~\ref{subsec:llp-constrained},
    \item Section~\ref{subsec:k-trans} then proposes $k$-transition predicates and proves (Theorem~\ref{thm:k-trans}) that Algorithm~\ref{alg:constrained-cuts} terminates in $O(k)$ rounds in the presence of strong advancement,
    \item  Next, Section~\ref{subsec:irreducibles} uses Algorithm~\ref{alg:constrained-cuts} to derive Algorithm~\ref{alg:irreducibles}, which computes the irreducibles of a sublattice of min-cuts satisfying regular predicates,
    \item With the irreducibles in hand, Section~\ref{subsec:enumeration} uses Birkhoff's theorem to propose an enumeration method in Algorithm~\ref{alg:enumeration} for min-cuts satisfying regular predicates,
    \item Finally, Section~\ref{subsec:general} uses Algorithm~\ref{alg:enumeration} and poset slicing \cite{garg2001slicing,Garg06-tcs} to obtain an exact search method for the case of general constraints in Algorithm~\ref{alg:slicing}, with improved complexity over exhaustive search.
\end{itemize}

\section{Related Work}\label{sec:related-work}

Our ability to apply lattice-linear predicate detection begins from the observation that the set of min-cuts forms a distributive lattice \cite{picard1980structure}.
An enumeration algorithm for min-cuts was presented in \cite{provan1996paradigm}.
Their method is based on partitioning the graph structure, and cleverly traversing the resulting decomposition.
We obtain a different enumeration algorithm by instead precomputing the irreducibles.
An advantage of our method is it is more robust towards being applied in other scenarios described through additional constraints, as it applies to any conjunction of regular predicates.
For example, \cite{provan1996paradigm} also present enumeration methods for uniformly directed cuts, which we describe via regular predicates in Section~\ref{subsec:prelim-mincut}.
So, we recover their enumeration method for uniformly directed cuts ``for free,'' as well any other constraint structures defined by regular predicates.
Other methods exist for enumerating min-cuts \cite{picard1980structure, ball1983calculating, gardner1985algorithm}, but have efficiency issues when applied to directed graphs (see \cite{provan1996paradigm}).
There are also methods for enumerating global min-cuts \cite{henzinger2020finding} and variants \cite{karger2016enumerating, beideman2023approximate} beyond the scope of this work.

Application of lattice-linear predicate detection towards combinatorial problems was introduced in \cite{DBLP:conf/spaa/Garg20} by building on algorithms for verifying distributed systems \cite{chase1998detection}.
A key feature is the ability to analyze and solve a large variety of combinatorial problems in a unifying way.
The assignment problem, scheduling, and stable marriage were investigated in \cite{DBLP:conf/spaa/Garg20}, while follow up works have applied lattice-linearity to dynamic programming \cite{Garg:ICDCN22}, minimum spanning trees \cite{AlvGar22}, the housing market problem \cite{garg2021lattice}, and multiplication and modulo computations \cite{gupta2025tolerance}.
Furthermore, applications to self-stabilizing and asynchronous distributed computing algorithms are also examined in \cite{gupta2021extending, gupta2024tolerance, gupta2025tolerance, gupta2024eventually}.
Finally, \cite{DBLP:conf/icdcn/GargS24} recently introduced another class of predicates, called \emph{equilevel predicates}, with relationships to lattice-linear predicates.

\section{Preliminaries}\label{sec:preliminaries}

Some conventions:
Given a set $X$ and $y \notin X$, we use $X + y \triangleq X \cup \{y\}$ to denote the union of a set with a singleton.
We do the same for difference, with $X - y \triangleq X\setminus \{y\}$.
We assume some familiarity with order theory, see \cite{davey,Gar:2015:bk}.

\subsection{Lattice-Linearity}\label{subsec:prelim-llp}
Examine a partially ordered set (poset) $\CCP = (\mathcal{X}, \leq)$.
We refer to the \emph{meet} (greatest lower bound) of $\mathbf{x}$ and $\mathbf{y}$ by $\mathbf{x} \sqcap \mathbf{y}$ and the \emph{join} (least upper bound) by $\mathbf{x} \sqcup \mathbf{y}$, when they exist.
A poset is a \emph{semilattice} whenever it is closed under one of join or meet, and a \emph{lattice} when closed under both.
A lattice $\CC$ is distributive if the join operation distributes over the meet operation in the algebraic sense.
For example, a \emph{ring of sets}, i.e. a set family closed under union and intersection (its join and meet, respectively), is a distributive lattice.
A \emph{sublattice} is a subset of the lattice closed under its meet and join operation.
A sublattice of a distributive lattice is always distributive.
A \emph{(join-)irreducible} $\mathbf{i} \in \CC$ is a lattice element such that $\mathbf{x}\sqcup\mathbf{y} = \mathbf{i}$ implies $\mathbf{x} = \mathbf{i}$ or $\mathbf{y} = \mathbf{i}$.
An \emph{(order) ideal} of the poset $\CCP$ is a set $X \subseteq \mathcal{X}$ such that
$y \in X$ and $z \leq y$ implies $z \in X$.
A \emph{filter} is the dual concept.
Birkhoff's representation theorem establishes an equivalence between distributive lattices and the ideals of the poset of its irreducibles.

\begin{theorem}[Birkhoff's Representation Theorem \cite{birkhoff1937rings}]
    Every distributive lattice $\CC$ is isomorphic (i.e. there exists an order-preserving bijection) to the ideals of the poset of its irreducibles ordered according to $\CC$.
\end{theorem}

We will use distributive lattices to model the search space of the problems examined in this paper.
For what follows, let $\CC$ be the lattice of all $n$-dimensional non-negative vectors of reals less than or equal to some fixed vector $U \in \mathbb{R}^n_+$, 
where the order is the component-wise $\leq$.
To model a computational problem over the lattice, we can define a predicate $B:\CC\to\{\textsf{T},\textsf{F}\}$ which is true on the set of solutions and false otherwise.
This gives a \emph{predicate detection} problem, where the task is to find a lattice element that satisfies $B$.
Whenever this predicate is \emph{lattice-linear} \cite{chase1998detection,DBLP:conf/spaa/Garg20}, this is simplified.
Specifically, imagine an algorithm that maintains a vector and takes iterative actions based on its value.
Then, each vector coordinate is a \emph{(local) state}, and these states compose together to define the current configuration at any step.
Given $G \in \CC$, the state $G(i)$ (or, equivalently, the index $i$) is \emph{forbidden} whenever any vector $H$ greater than $G$ with $G(i) = H(i)$ is such that $B$ is false in $H$.
A predicate $B$ is {\em lattice-linear}
 if for any $G \in \CC$, $B$ is false in $G$ implies that $G$ contains a
forbidden state.
Intuitively, this means that any algorithm that updates its configuration $G$ simply by increasing the coordinates \emph{must} increase the entries on the forbidden indices in any step if it hopes to find a suitable solution.

This leads to an iterative algorithm, the \emph{LLP algorithm}, defined by the following rule: In each round, if there exists a forbidden state in the current candidate solution $G$, then increase some forbidden state until it is no longer forbidden (i.e. \emph{advanced}).
This continues until there are no more forbidden states or until the top of the lattice is reached.
In the former case, the definition of lattice-linearity certifies that we have found a solution satisfying $B$.
 In the latter case, there is no configuration greater than or equal to $G_0$ satisfying $B$.
See \cite{DBLP:conf/spaa/Garg20} for a more detailed specification and discussion.
Essentially, detecting lattice-linear predicates simplifies to the task of finding forbidden states.
Furthermore, the LLP method is particularly amenable to parallel algorithm design as one can advance multiple forbidden states independently in parallel without harming correctness, see \cite{DBLP:conf/spaa/Garg20} for details.
We use this fact implicitly in our work.
Finally, the following is useful for the analysis of lattice-linear predicates.

\begin{lemma}[\cite{chase1998detection,DBLP:conf/spaa/Garg20}]\label{lem:basic-LLP} 
    Let $B_1,B_2:\CC\to\{\textsf{T},\textsf{F}\}$ be predicates on a distributive lattice $\CC$. If $B_1$ and $B_2$ are lattice-linear, then $B_1 \wedge B_2$ is also lattice-linear.
\end{lemma}

Lattice-linearity is intimately related to the structure of a predicate's satisfying preimage, as it is equivalent to these elements possessing a semilattice structure.

\begin{proposition}[\cite{chase1998detection,DBLP:conf/spaa/Garg20}]\label{prop:semilattice-rep}
    Let $B:\CC\to\{\textsf{T},\textsf{F}\}$ be a predicate on a distributive lattice $\mathscr{L}$.
    Then, $B$ is lattice-linear if and only if $B(G)\wedge B(H) \implies B(G\sqcap H),$
    i.e. the set of elements satisfying $B$ is a meet-subsemilattice of $\CC$.
\end{proposition}

Now, a state $G(i)$ is \emph{dual-forbidden} whenever $H\leq G$ and $H(i) = G(i)$ implies $\neg B(H)$, and so $B$ is \emph{dual lattice-linear} whenever $\neg B(G)$ implies the existence of a dual-forbidden state.
Hence, a dual lattice-linear predicate is lattice-linear over the dual lattice (given by inverting the ordering relation),
and so Proposition~\ref{prop:semilattice-rep} implies the satisfying elements of $B$ make a semilattice with the join operation of $\CC$.
Keeping this in mind, call a predicate \emph{regular} whenever it is lattice-linear and dual lattice-linear.
Thus, the satisfying elements of any regular predicate $B$ form a sublattice of $\CC$.

As a technical note,  unless $\RP = \NP$ computing forbidden states cannot generally admit a polynomial-time algorithm \cite{kashyap2005intractability}.
So, one assumes the \emph{efficient advancement property}, stating that there exists an efficient procedure to identify and advance some forbidden state.
This is not a strong assumption; all previously examined applications satisfy this property.
In our work, we examine algorithms seeking to satisfy conjunctions of lattice-linear predicates.
In this setting, a strengthened definition can improve complexity bounds in parallel computation settings.

\begin{definition}[Strong Advancement]
    Let $B:\mathscr{L}\to\{\textsf{T},\textsf{F}\}$ be a lattice-linear predicate.
    Then, $B$ has \emph{strong $T$-advancement} if for any $G\in\mathscr{L}$, \emph{every} forbidden state is identified and advanced to a least state in which it is not forbidden in $T$ time.
\end{definition}

This definition says that it is possible to advance \emph{all} forbidden states of any $G$ to satisfy the predicate in time $T$ (whereas efficient advancement only asserts this guarantee for \emph{some} forbidden states).
Take note that every lattice-linear predicate has strong $T$-advancement if one makes $T$ large enough, as the forbidden states and their needed advancements can be identified by computing the minimum satisfying element above a given $G$.
Strong advancement is more notable when examining conjunctions of predicates, as individual $T$-strong advancement guarantees on the conjuncts may be destroyed by conflicting advancements in each step.
Still, in the presence of other properties, strong advancement can imply improved complexity bounds for LLP algorithms operating on conjunctions (see Theorem~\ref{thm:k-trans}).
We show examples of lattice-linear predicates relevant to min-cut computations in the next section.

\subsection{Max-Flows and Min-Cuts}\label{subsec:prelim-mincut}
Let $(V,E)$ be a directed graph over $n+2$ vertices and $m$ edges, with two distinguished vertices $s$ and $t$ called the \emph{source} and \emph{sink}.
We assume that $(V,E)$ is a \emph{flow-network} in the sense that $s$ and $t$ are \emph{terminal}.
This is without loss of generality, as any directed graph can be transformed to satisfy this condition while preserving the cut capacities and flow values by appending terminal vertices to $s$ and $t$ to act as the source and sink.
An \emph{$(s,t)$-cut} $(S,T) \in V\times V$ is a partition of the vertices into disjoint subsets such that $s \in S$ and $t \in T$.
We denote the set of $(s,t)$-cuts by $\mathcal{C}(s,t)$.
For any $(s,t)$-cut $(S,T)$, its \emph{cut-set} is the set of edges directed from $S$ to $T$.
And, given a \emph{capacity function} $c:V\to\mathbb{R}_+$, the capacity of a cut is the sum $\sum_{(u,v)\in (S\times T) \cap E} c(u,v)$ of the capacities of the edges in its cut-set.
Then, a $(s,t)$-cut is \emph{minimum} if its capacity is less than or equal to that of any other $(s,t)$-cut.
In what follows, we call such a cut a \emph{min-cut}.

\begin{figure}[t]
    \centering
    \begin{subfigure}{\linewidth} 
        \centering
        \scalebox{.75}{
    \begin{tikzpicture}
        \node[draw, circle] (s) at (0,0) {$s$};
        \node[draw, circle] (a) at (1.5, 1.5) {$a$};
        \node[draw, circle] (b) at (2.25, 0) {$b$};
        \node[draw, circle] (c) at (1.5, -1.5) {$c$};

        \node[draw, circle] (d) at (3.75, 1) {$d$};
        \node[draw, circle] (e) at (3.75, -1) {$e$};
        \node[draw, circle] (t) at (5.5, 0) {$t$};

        \draw[->] (s) to node[midway, above left] {$2$} (a);
        \draw[->] (s) to node[midway, above] {$1$} (b);
        \draw[->] (s) to node[midway, below left] {$1$} (c);

        \draw[->] (b) to node[midway, right] {$1$} (a);

        \draw[->] (a) to node[midway, above] {$2$} (d);
        \draw[->] (b) to node[midway, below right] {$1$} (d);
        \draw[->] (c) to node[midway, below] {$1$} (e);
        
        \draw[->] (d) to node[midway, above right] {$2$} (t);
        \draw[->] (e) to node[midway, below right] {$1$} (t);
\end{tikzpicture}}
    \caption{A $(s,t)$-flow network.}
    \end{subfigure}
    \begin{subfigure}{\linewidth}
        \centering
        \scalebox{.8}{
            \begin{tikzpicture}
            \tikzstyle{S}=[rectangle, draw=black, rounded corners=5pt]
            \tikzstyle{I}=[rectangle, double, thick, draw=black, rounded corners=5pt]

            \node[S] (s) {$s\mid a\;b\;c\;d\;e\;t$};
            \node[I] (a) [above left=of s,yshift=-10pt,xshift=54pt] {$s\;a \mid b\;c\;d\;e\;t$};
            \node[I] (c) [above right=of s,yshift=-10pt,xshift=-54pt] {$s\;c\mid a\;b\;d\;e\;t$};
            \node[S] (ac) [above=of a,yshift=-10pt] {$s\;a\;c\mid b\;d\;e\;t$};
            \node[I] (ad) [left=of ac,xshift=19pt] {$s\;a\;d\mid b\;c\;e\;t$};
            \node[I] (ab) [above=of c,yshift=-10pt] {$s\;a\;b\mid c\;d\;e\;t$};
            \node[I] (ce) [right=of ab,xshift=-19pt] {$s\;c\;e\mid a\;b\;d\;t$};
            \node[S] (abc) [above=of ac, xshift=-20pt,yshift=-10pt] {$s\;a\;b\;c\mid d\;e\;t$};
            \node[S] (abd) [above=of ab, xshift=-20pt,yshift=-10pt] {$s\;a\;b\;d\mid c\;e\;t$};
            \node[S] (ace) [above=of ce, xshift=-20pt,yshift=-10pt] {$s\;a\;c\;e\mid b\;d\;t$};
            \node[S] (acde) [above=of abc, xshift=-20pt,yshift=-10pt] {$s\;a\;c\;d\;e\mid b\;t$};
            \node[S] (abcd) [above=of abd, xshift=-17pt,yshift=-10pt] {$s\;a\;b\;c\;d\mid e\;t$};
            \node[S] (abce) [above=of ace, xshift=-20pt,yshift=-10pt] {$s\;a\;b\;c\;e\mid d\;t$};
            \node[S] (top) [above=of s, yshift=130pt] {$s\;a\;b\;c\;d\;e\mid t$};
            \node (inv) [above=of s, yshift=150pt] {$~$};

            \foreach \from/\to in {s/a, s/c, a/ad, a/ac, a/ab, c/ce, c/ac, ad/abd, ad/acde, ac/abc, ac/ace, ab/abc, ab/abd, ce/ace, abc/abcd, abc/abce, abd/abcd, ace/acde, ace/abce, acde/top, abcd/top, abce/top}
            \draw (\from) -- (\to);
        \end{tikzpicture}
    }
    \caption{Lattice of min-cuts.}
    \end{subfigure}
    \caption{Above is a flow network and Hasse diagram of its lattice of min-cuts.
        We identify the irreducibles by the nodes with the double borders in the Hasse diagram. 
    }\label{fig:irreducibles}
\end{figure}
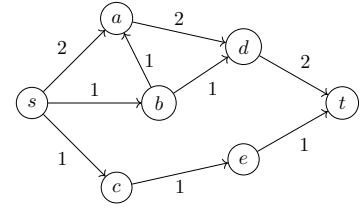
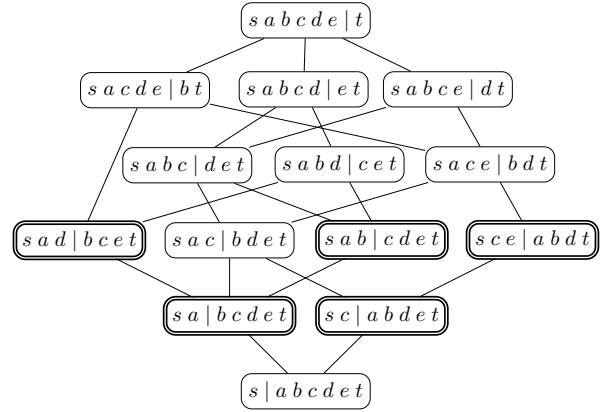

We will use lattice-linear predicate detection to identify min-cuts satisfying additional constraints.
To model this, note that cuts can be ordered by \emph{refinement}, which is to say $(S,T)\leq(S',T')$ if and only if $S\subseteq S'$.
It can be seen then that $(\mathcal{C}(s,t),\leq)$ is in isomorphism with a Boolean algebra over $V\setminus\{s,t\}$ by way of the mapping $(S,T)\mapsto S-s$.
Thus it is a distributive lattice with,
\begin{align*}
    (S,T)\sqcup (S', T') = (S\cup S', V\setminus (S\cup S')),
\end{align*}
and,
\[
(S,T)\sqcap (S',T') = (S\cap S', V\setminus (S\cap S')).
\]
It is known that the min-cuts form a sublattice of this construction \cite{picard1980structure,queyranne1998minimizing}, see Figure~\ref{fig:irreducibles} for example.
We will apply the LLP method to $(\mathcal{C}(s,t),\leq)$.
Here, a vertex $u$ is in some forbidden state, henceforth just forbidden, for the cut $(S,T)$ if there is no greater refining cut satisfying a given predicate such that $u$ remains in the same side of the bipartition.
Then, advancement to a greater refining cut corresponds to moving $u$ from $T$ into $S$.
Interesting examples of lattice-linear predicates follow.

\begin{example}[Implicational Cuts]
    Let $X\subseteq V\setminus\{s,t\}$ be some set of vertices and $u \in V\setminus (X\cup\{s,t\})$.
    Define $B_{X\Rightarrow u}(S,T)\triangleq X\subseteq S \implies u \in S$.
    Observe, $B_{X\Rightarrow u}(S,T)$ is false whenever $X \subseteq S$ but $u\notin S$.
    In this case, $u$ is forbidden and must be advanced by updating to $(S+u,T-u)$.
\end{example}

Note that the implicational cut predicate $B_{X\Rightarrow u}$ has a strong $O(1)$-advancement with $O(n)$ work: One checks each of $X\cup u$ for membership in $S$ in parallel.
Then, if $X\subseteq S$ but $u\notin S$, $u$ is forbidden and advanced by updating $(S,T)$ to $(S+u, T-u)$.
Otherwise, no vertex is forbidden.

We also describe \emph{uniformly directed $(s,t)$-cuts}, i.e. $(S,T)$ with no edge directed from $T$ to $S$, using regular predicates.
Such cuts are examined in prior work \cite{provan1989exact,provan1996paradigm,sigal1980stochastic,avramidis1996integrated}

\begin{example}[Uniformly Directed Cuts]\label{ex:udc}
    For this example, define $B_{(u,v)}(S,T) \triangleq \neg(u \in S \wedge v \in T)$ for each edge $(u,v)\in E$.
    Then, recalling that we assume that $s$ possesses no incoming edges and $t$ possesses no outgoing edges, we see that an $(s,t)$-cut $(S, T)$ is uniformly directed if and only if $\bigwedge_{(u,v)\in E} B_{(u,v)}(S,T)$. 
    However, by applying De Morgan's law and the definition of implication, we also have,
    \begin{multline*}
            \neg(u \in S \wedge v \in T)\equiv u \notin S \vee v \notin T,\\\equiv u \notin S \vee v \in S\equiv \left( \{u\} \subseteq S \implies v \in S \right).
    \end{multline*}
    Therefore, $\bigwedge_{(u,v)\in E} B_{(u,v)} = \bigwedge_{(u,v)\in E} B_{\{u\}\Rightarrow v}$,
    and lattice-linearity follows from Lemma~\ref{lem:basic-LLP}.
    Moreover, the only way $B_{\{u\}\Rightarrow v}$ is not satisfied is when $v \in T$ and $u \in S$.
    Then, $u$ is clearly dual-forbidden since any cut refined by $(S,T)$ will contain $v\in T$.
    So, each $B_{\{u\}\Rightarrow v}$ is also dual lattice-linear, implying that the uniformly directed cuts predicate is regular by Lemma~\ref{lem:basic-LLP} and Proposition~\ref{prop:semilattice-rep}.
\end{example}

In the last example, we described uniformly directed cuts by a conjunction of $m$ implicational cut predicates. 
This representation then has strong $O(1)$-advancement with $O(nm)$ work.
Now, let $B_\text{MC}$ be the predicate which is true on the cut $(S,T)$ if and only if $(S,T)$ is a min-cut (a symbolic definition is not important and omitted for space).
In the next section, we show that $B_\text{MC}$ is regular
by using \emph{max-flows}.
Solving a max-flow efficiently can be done via any of \cite{edmonds1972theoretical,dinitz2006dinitz,goldberg1988new,orlin2013max}, and at the time of writing it appears that the breakthrough $O\left(m^{1+o(1)}\right)$-time achieved by \cite{chen2025maximum} is the best known asymptotic complexity.
The parallel algorithms literature is more sparse; a method for computing exact flows in sublinear parallel time does not seem known, but there are approximation schemes terminating in polylogarithmic time \cite{serna1991tight,agarwal2024parallel}.
We parameterize our time and work complexities in anticipation of future improvements.
Specifically, let $T_\text{MF}$ and $W_\text{MF}$ be the time and work complexities of a parallel max-flow computation over $n$ vertices and $m$ edges.

\section{Min-Cuts are Regular}
\label{sec:mincutLattice}
It is known that the set of min-cuts forms a distributive lattice \cite{picard1980structure,queyranne1998minimizing}, and so regularity itself follows from Proposition~\ref{prop:semilattice-rep}.
However, this does not have immediate algorithmic use, as it does not imply an efficient search method for forbidden vertices.
This is the purpose of this section; all forbidden vertices can be found and advanced in constant parallel time if given a max-flow solution.

Our characterization of forbidden vertices begins with the following:
A cut is a min-cut if and only if there is no unsaturated edge or backward-directed edge with nonzero flow crossing the cut.
Prior work has identified similar statements (as it is an immediate result of linear programming complementary slackness, see \cite{picard1980structure} or Ch. 8 of \cite{korte2008combinatorial} for translation into the language of augmenting paths), so we state without proof.

\begin{lemma}[Folklore]\label{lem:forbidden-characterization}
    Let $(V,E)$ be a $(s,t)$-flow network and $(S,T)$ a $(s,t)$-cut.
    Consider a max-flow $f^* : E \to \mathbb{R}_+$.
    Then, $(S, T)$ is a min-cut if and only if its cut-set is saturated by $f^*$, and edges directed from $T$ to $S$ have flow equal to zero. Symbolically, this is,
    \begin{align*}
        (u,v) \in (S\times T)\cap E\implies f^*(u,v) = c(u,v),
    \end{align*}
    and,
    \[
        (v,u) \in (T \times S) \cap E \implies f^*(v,u) = 0.
    \]
\end{lemma}

As a result, the presence of an unsaturated forward edge or backwards edge with nonzero flow crossing the cut indicates that the least refining min-cut above some (non-minimum) cut cannot be reached through advancements until such an edge no longer crosses the cut.
In the case where there is an unsaturated forward edge $(u,v)$, it follows that $v$ must be added to the $s$-side of the cut.
Likewise, the same is true if $v$ is the start point of a backwards edge with nonzero flow, and so in either case it follows that $v$ is forbidden.
This can be tested for all edges in parallel to find the forbidden vertices with constant overhead, see Algorithm~\ref{alg:forbidden}.
Although we do not show it explicitly, the dual logic applies for the dual-forbidden vertices.
The soundness of this approach is now proven.

\begin{algorithm}[t]
    \caption{Searching for Forbidden Vertices}\label{alg:forbidden}
\DontPrintSemicolon
\SetKwRepeat{Do}{do}{while}
\SetKwInOut{KwInput}{input~}
\SetKwInOut{KwOutput}{output~}
\KwInput{$(s,t)$-network $(V, E)$, capacities $c:E\to\mathbb{R}_+$, max-flow $f^*$, and $(S, T)\in\mathcal{C}(s,t)$.}
\KwOutput{All forbidden vertices of $(S,T)$.}

\vspace{\baselineskip}

$X \gets \varnothing$\;
\;
\ForEach{$(u,v) \in E$ in parallel}
{
    \If{$(u, v) \in T \times S\text{ and } f^*(v,u) \neq 0$}
    {
        $X\gets X+u$
    }
    \ElseIf{$(u,v) \in S\times T\text{ and } f^*(u,v) < c(u,v)$}
    {
        $X \gets X+v$
    }
}
\;
\Return $X$
\end{algorithm}
\begin{theorem}\label{thm:min-cuts-ll}
    Let $(V, E)$ be a $(s,t)$-flow network.
    Then, the min-cut predicate $B_\textup{MC}$ is regular with strong $O(1)$-advancement over $(\mathcal{C}(s,t),\leq)$ if given any max-flow.
\end{theorem}
\begin{proof}
    Fix a $(s,t)$-cut $(S, T)$, and assume that $(S, T)$ is not a min-cut.
    Let $f^*$ be any max-flow.
    By Lemma~\ref{lem:forbidden-characterization}, there exists an edge $(u,v) \in E$ crossing the cut such that $(u,v)$ is either directed from $T$ to $S$ with positive flow from $f^*$, or is directed from $S$ to $T$ and not saturated by $f^*$.
    We claim that $u$ is forbidden in the former case, and $v$ is forbidden in the latter.
    Suppose $(u,v)$ is directed from $T$ and has positive flow, and
    consider any $(S',T') \geq (S,T)$ such that $u \notin S'$.
    Notice then, $S'\supseteq S$ implies $v \in S$, and so $(u,v)$ also crosses $(S', T')$.
    So, $(u,v)$ is an edge directed from $T'$ to $S'$ with positive flow, implying $(S',T')$ is not a min-cut by Lemma~\ref{lem:forbidden-characterization}.
    Hence $u$ is forbidden.
    Now, assume $(u,v)$ is directed from $S$ to $T$ but is not saturated,
    and examine $(S',T')\geq(S,T)$ with $v \notin S'$.
    Similar to before, our assumptions make $u \in S'$, and so $(u,v)$ crosses $(S',T')$.
    Hence, $(u,v)$ is an edge directed from $S'$ to $T'$ not saturated by $f^*$, so $v$ is forbidden.

    Therefore, finding a forbidden vertices is as easy as computing a max-flow and enumerating the edges.
    Testing an edge is done in constant time after precomputing a max-flow, and this can be done in one parallel step. 
    The symmetric argument shows dual-forbidden vertices (which would be the end of some edge directed from $T$ to $S$ or the beginning of an unsaturated edge directed from $S$ to $T$) can be computed in much the same way.
    Thus, $B_\textup{MC}$ is regular.
\end{proof}

Though similar in nature to typical algorithmic reductions between min-cuts and max-flows, we require the structure given by Lemma~\ref{lem:forbidden-characterization} to diagnose the specific vertices required to advance towards a refining min-cut from \emph{any} graph cut.
Such is necessary for LLP computations in the presence of additional constraints, examined in the next section.

\section{Constraint Satisfying Minimum Cuts}\label{sec:constraint-satisfying-mincuts}

Now, examine the task of identifying a min-cut satisfying a set of constraints.
Although \NP-hard in general, by combining the results of the last section with the fact that lattice-linear predicates are closed under conjunction (see Lemma~\ref{lem:basic-LLP}.2) we can solve this problem via the LLP algorithm whenever the constraints are defined by lattice-linear predicates.
We show this in Section~\ref{subsec:llp-constrained}, and introduce $k$-transition predicates and their application to parallel complexity analysis in Section~\ref{subsec:k-trans}.
With that said, define the \emph{constraint satisfying min-cuts} decision problem as follows.
\begin{quote}
    {\bf Instance}: $(s,t)$-flow network $(V,E)$, capacities $c:E\to\mathbb{Z}$, and $B: \mathcal{C}(s,t)\to\{\textup{\textsf{T}},\textup{\textsf{F}}\}$.

    \vspace{0.125\baselineskip}

    {\bf Question:} Is there a min-cut with $B(S,T) = \textsf{T}$?
\end{quote}

The reader should note that this problem is distinct from the optimization variant asking for the cut of minimum capacity among those satisfying certain constraints.
Instead, this problem asks for a cut satisfying the condition of being a min-cut in addition to other constraints defined by $B$, i.e. the predicate $B_\text{MC} \wedge B$.
By reduction from satisfiability this is \NP-hard.

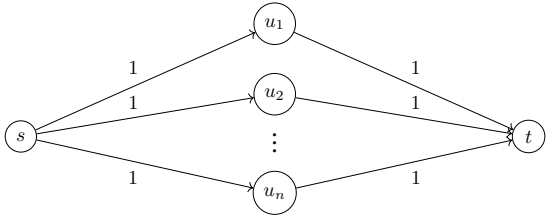
\begin{figure}[t]
\centering
\caption{The graph used in the proof of Theorem~\ref{thm:hardness}.
}\label{fig:hardness}
\scalebox{.75}{\begin{tikzpicture}

    \node[draw, circle] (s) at (0,0) {$s$};

\node[draw, circle] (u1) at (4.5,2) {$u_1$};
\node[draw, circle] (u2) at (4.5,0.75) {$u_2$};
\node at (4.5,-0) {\Large$\vdots$};
\node[draw, circle] (un) at (4.5,-1) {$u_n$};

\node[draw, circle] (t) at (9,0) {$t$};

\draw[->] (s) to node[midway, above left] {$1$} (u1);
\draw[->] (s) to node[midway, above left] {$1$} (u2);
\draw[->] (s) to node[midway, below left] {$1$} (un);

\draw[->] (u1) to node[midway, above right] {$1$} (t);
\draw[->] (u2) to node[midway, above right] {$1$} (t);
\draw[->] (un) to node[midway, below right] {$1$} (t);

\end{tikzpicture}}
\end{figure}

\begin{theorem}\label{thm:hardness}
    The constraint satisfying min-cuts problem is \NP-complete.
\end{theorem}

\begin{proof}
First, note that the problem is in \NP. Given an $(s,t)$-cut, one can easily verify that it is indeed a min-cut and satisfies $B$.
For hardness, reduce from satisfiability: 
Examine a Boolean formula $B'$ over $n$ variables $y_1, \ldots, y_n$.
We construct a directed graph with source $s$ and sink $t$ over $n+2$ vertices by first
labelling the non-terminal vertices as $u_1,u_2, \ldots u_n$ and placing an edge from $s$ to $u_i$ with capacity $1$ for all $i\in\{1,\ldots, n\}$. 
Then, another edge is placed from $u_i$ to $t$ with capacity $1$ for all $i\in\{1,\ldots, n\}$ as well.
See Figure~\ref{fig:hardness}.
Observe, each possible min-cuts corresponds exactly to a subset of $\{u_1,\ldots,u_n\}$.
So, define $B: \mathcal{C}(s,t)\to\{\textup{\textsf{T}},\textup{\textsf{F}}\}$ as $B(S,T) \triangleq B'(\{y_i\mid u_i \in S\})$.
Hence, there exists a min-cut satisfying $B$ if and only if $B'$ is satisfiable.
\end{proof}

\subsection{Solution by an LLP Algorithm}\label{subsec:llp-constrained}
\begin{algorithm}[t]
    \caption{Min-Cuts with Lattice-Linear Predicates}\label{alg:constrained-cuts}
\DontPrintSemicolon
\SetKwRepeat{Do}{do}{while}
\SetKwInOut{KwInput}{input~}
\SetKwInOut{KwOutput}{output~}
\KwInput{$(s,t)$-network $(V, E)$, capacities $c:E\to\mathbb{Z}_+$, lattice-linear $B:(\mathcal{C}(s,t),\leq)\to\{\textsf{T},\textsf{F}\}$.}
\KwOutput{The least min-cut $(S,T)\in(\mathcal{C}(s,t),\leq)$ satisfying $B$ if one exists, and $\bot$ otherwise.}

\vspace{\baselineskip}

$(S,T)\leftarrow (V-t, \{t\})$\;
\;
\Do{$F_\text{MC}\cup F_B \neq \varnothing \text{ and } (S,T)\neq (\{s\},V-s)$}
{
    $F_\text{MC} \gets \texttt{searchForbidden}((S,T), B_\text{MC})$\;
    $F_B \gets \texttt{searchForbidden}((S,T), B)$\;
    $(S,T) \gets \left(S \cup F_\text{MC}\cup F_B, T\setminus (F_\text{MC}\cup F_B)\right)$
}
\;
\lIf{$F_\text{MC}\cup F_B\neq\varnothing$}
{
    \Return $\bot$
}
\lElse
{
    \Return $(X+s, V\setminus X + t)$
}
\end{algorithm}

By Lemma~\ref{lem:basic-LLP} and Theorem~\ref{thm:min-cuts-ll}, the conjunction of $B_\text{MC}$ with a lattice-linear predicate is lattice-linear.
Therefore, in spite of the general hardness, constraint satisfying min-cuts can be solved efficiently via the LLP method if the constraints encode lattice-linear predicates with efficient advancement.
Specifically, we instantiate this in Algorithm~\ref{alg:constrained-cuts} for a lattice-linear predicate $B:(\mathcal{C}(s,t),\leq)\to\{\textsf{T},\textsf{F}\}$.
The algorithm proceeds in rounds,
wherein each round the algorithm searches for forbidden vertices induced by a candidate solution $(S,T)$ against the predicates $B_\text{MC}$ and $B$, distinguished by $F_\text{MC}$ and $F_B$ respectively.
The identified forbidden vertices are advanced by assigning $(S,T)$ to $(S\cup F_\text{MC}\cup F_B,T\setminus(F_\text{MC}\cup F_B))$.
This repeats until are no more forbidden vertices to be found, or until $(S,T)$ equals the top cut.
In the former case, by lattice-linearity $(S,T)$ is a min-cut satisfying the constraints, and so $(S,T)$ is returned.
In the latter case, where $(S,T)$ equals the top cut and there exist forbidden vertices, it follows that no solution can exist and 
so a null value $\bot$ is returned.

Observe, all forbidden vertices induced by $B_\text{MC}$ are computed by Algorithm~\ref{alg:forbidden} in $O(1)$ time and $O(m)$ work after computing a max-flow.
This max-flow can be precomputed once and shared for future searches for forbidden vertices against $B_{MC}$.
Now, let $T_B$ and $W_B$ be the worst-case time and work complexities of computing a forbidden vertex induced by $B$.
There can only be at most $n$ rounds (as the loop must end if the cut is not refined), so we see that the algorithm terminates in $O(T_\text{MF} + nT_B)$ time with $O(W_\text{MF} + n(m + T_B))$ work.

\subsection{$k$-Transition Predicates under Strong Advancement}\label{subsec:k-trans}

For many predicates, Algorithm~\ref{alg:constrained-cuts} will perform better than the coarse guarantees given in the last section.
For example, consider the simple implicational cut predicate $B_{\{u\}\Rightarrow v}$ for fixed vertices $u,v\in V$.
Algorithm~\ref{alg:constrained-cuts} will terminate in 3 rounds:
In the first, $(S,T)$ is advanced to the least min-cut by refinement.
If $u \notin S$ or $v \in S$, then $B_{\{u\}\Rightarrow v}$ is satisfied and the algorithm terminates.
Otherwise, $(S,T)$ is advanced to the least refining cut containing both $u$ and $v$ in $S$.
If the result is a min-cut then the algorithm terminates, and if not, then the algorithm advances on the forbidden vertices induced by $B_\text{MC}$.
The algorithm then finds a solution satisfying $B_\text{MC}\wedge B_{\{u\}\Rightarrow v}$, or that none exist, in the third round.

Underlying this is $B_{\{u\}\Rightarrow v}$ is limited in the number of times it can flip between true and false throughout the computation.
To abstract this idea, we introduce \emph{$k$-transithon predicates}.

\begin{definition}\label{defn:k-transition}
    Let $\mathscr{L}$ be a distributive lattice.
    Then, a predicate $B:\mathscr{L}\to\{\textup{\textsf{T}},\textup{\textsf{F}}\}$ is a \emph{$k$-transition predicate} if and only if, for every chain $\mathbf{x}_1\leq\mathbf{x}_2\leq\ldots$ in $\mathscr{L}$, the count of indices $i$ such that $B(\mathbf{x}_i) \neq B(\mathbf{x}_{i+1})$ is at most $k$.
\end{definition}

Implicational cuts are $2$-transition predicates as their satisfying preimage is the union of an order ideal and filter.
Furthermore, \emph{solitary predicates} \cite{DBLP:conf/icdcn/GargS24}, being predicates satisfied by a unique assignment, are 2-transition predicates, while \emph{stable predicates} \cite{ChanLamp:Snap}, whose satisfying preimage is an order filter, are 1-transition predicates.
Generally, because the length of any strict chain in $(\mathcal{C}(s,t),\leq)$ is bounded by the number of vertices, every predicate over the lattice of $(s,t)$-cuts is a $n$-transition predicate.
With this and strong advancement, we bound the number of rounds in Algorithm~\ref{alg:constrained-cuts} by $O(k)$ for improved complexity guarantees.

\begin{theorem}\label{thm:k-trans}
    Let $(V,E)$ be a $(s,t)$-flow network, and $B:(\mathcal{C}(s,t),\leq)\to\{\textsf{T},\textsf{F}\}$ a lattice-linear predicate.
    Suppose $B$ is a $k$-transition predicate with strong $T_B$-advancement using $W_B$ work. Then, Algorithm~\ref{alg:constrained-cuts} terminates in $O(T_\text{MF} + kT_B)$ parallel time with $O(W_\text{MF} + k(m + W_B))$ work.
\end{theorem}
\begin{proof}
    We argue Algorithm~\ref{alg:constrained-cuts} terminates in $O(k)$ rounds.
    Then, the time bound follows from the overhead of computing forbidden vertices for $B_\text{MC}$ being constant after precomputing a max-flow, while work follows from noting Algorithm~\ref{alg:forbidden} performs tests using constant operations for each edge.

    Let $\left(S^{(t)},T^{(t)}\right)$ be the value of $(S,T)$ at the end of the $t^\text{th}$ round, with $\left(S^{(0)},T^{(0)}\right)$ the bottom cut by convention, while $F_\text{MC}^{(t)}$ and $F_B^{(t)}$ are the forbidden vertices corresponding to $B_\text{MC}$ and $B$, respectively.
 Suppose that the computation terminates in $\eta$ rounds.
    By strong advancement,
    \begin{equation}\label{eq:k}
        B\left(S^{(t-1)}\cup F_B^{(t)},T^{(t-1)}\setminus F_B^{(t)}\right) = \textsf{T},
    \end{equation}
    for all $t \in \{1,\ldots,\eta\}$.
    And, by the update rules in Algorithm~\ref{alg:constrained-cuts}, we also observe the chain,
    \begin{multline*}
        \left(S^{(0)}\cup F_B^{(1)},T^{(0)}\setminus F_B^{(1)}\right) \leq \left(S^{(1)},T^{(1)}\right),
        \leq \ldots,\\
        \leq \left(S^{(\eta -1)}\cup F_B^{(\eta)},T^{(\eta-1)}\setminus F_B^{(\eta)}\right) \leq \left(S^{(\eta)},T^{(\eta)}\right).
    \end{multline*}
    By assumption, there are at most $k$ transitions between true and false (and vise versa) when applying $B$ to each member of this chain.
    Now, if $F^{(t)}_B \neq \varnothing$ then $B\left(S^{(t-1)},T^{(t-1)}\right)$ is false.
    So, by applying this insight and Eq.~\ref{eq:k} to our chain,
    \[
        \left|\left\{t \in \{1,\ldots,\eta\}\Bigm\vert F^{(t)}_B\neq\varnothing\right\}\right| \leq \frac{k}{2}.
    \]
    Furthermore, $F^{(t)}_B = \varnothing$ implies $\left(S^{(t)},T^{(t)}\right)$ is a min-cut by the strong advancement of $B_\text{MC}$.
    Then, unless $\left(S^{(t)},T^{(t)}\right)$ satisfies $B$ as well, in which case $i = \eta$ as the computation terminates once there are no more forbidden vertices to be found, we will have $F^{(t+1)}_B \neq \varnothing$.
    It follows that,
    \begin{multline*}
        \left|\left\{t \in \{1,\ldots,\eta\}\Bigm\vert F^{(t)}_B = \varnothing\right\}\right|,\\ \leq \left|\left\{t \in \{1,\ldots,\eta\}\Bigm\vert F^{(t)}_B\neq\varnothing\right\}\right| + 1 \leq \frac{k}{2} + 1.
    \end{multline*}
    So $\eta \leq k + 1$, and Algorithm~\ref{alg:constrained-cuts} terminates in $O(k)$ rounds.
\end{proof}

One can also bound the impact of conjunction on $k$.

\begin{lemma}\label{lem:k-trans-conj}
    Let $\mathscr{L}$ be a distributive lattice and $B_1,B_2:\mathscr{L}\to\{\textsf{T},\textsf{F}\}$ predicates.
    Suppose $B_1$ and $B_2$ are $k_1$ and $k_2$-transition predicates, respectively.
    Then $B_1\wedge B_2$ is a $(k_1+k_2)$-transition predicate.
\end{lemma}
\begin{proof}
    Fix any chain $\mathbf{x}_1\leq \mathbf{x}_2\leq \ldots$, and examine an $i$ with $(B_1\wedge B_2)(\mathbf{x}_i) \neq (B_1\wedge B_2)(\mathbf{x}_{i+1})$.
    Assume $(B_1\wedge B_2)(\mathbf{x}_i)$ is false.
    Hence, at least one of $B_1(\mathbf{x}_i)$ or $B_2(\mathbf{x}_i)$ is false, while both $B_1(\mathbf{x}_{i+1})$ and $B_2(\mathbf{x}_{i+1})$ are true.
    So, at least one of $B_1$ or $B_2$ transition from true to false between $\mathbf{x}_i$ and $\mathbf{x}_{i+1}$.
    The case where $(B_1\wedge B_2)(\mathbf{x}_i)$ is true follows from the symmetric argument.
    Therefore, $B_1\wedge B_2$ is a $(k_1+k_2)$-transition predicate.
\end{proof}

Unfortunately, there exist predicates eluding a description as a $k$-transition predicate for a useful value of $k$.
For example, uniformly directed cuts are a $\Theta(n)$-transition predicate:
Consider a graph over the vertices $s,v_1,\ldots,v_n,t$ with,
$$E = \{(v_i, v_{i+1})\mid i < n\}\cup\{(s,v_1),(v_n,t)\}.$$
Then define,
\begin{equation*}
    X_i \triangleq \begin{cases}
        \{s,v_2\},\quad&\text{if }i = 1,\\
        \{s,v_1,v_2\},\quad&\text{if } i = 2,\\
        \{s,v_1,\ldots,v_i\}\quad&\text{if } i>2\text{ and }i\text{ is even},\\
        \{s,v_1,\ldots,v_{i-1},v_{i+1}\},\quad&\text{otherwise},
    \end{cases}
\end{equation*}
and observe the strict chain $(X_1,V\setminus X_1) \leq \ldots \leq (X_n,V\setminus X_n)$ of $(s,t)$-cuts. By our construction, it can be seen that $(X_i,V\setminus X_i)$ is uniformly directed if and only if $i$ is even, and so any predicate describing uniformly directed cuts transitions between true and false $n-1$ times on this chain.

\section{Applications of the LLP Method}\label{sec:applic}

Now we apply our method to certain tasks.
First, in Section~\ref{subsec:irreducibles} we model the irreducible min-cuts as lattice-linear predicates with efficient advancement.
Algorithm~\ref{alg:constrained-cuts} can then be applied to construct a succinct representation of the lattice of min-cuts, and using Theorem~\ref{thm:k-trans} we show this is constructed in parallel time equal to that of computing a max-flow.
Moreover, these ideas extend to a conjunction with any regular predicate, allowing us to use this representation to enumerate min-cuts satisfying regular predicates in Section~\ref{subsec:enumeration}.
This of course gives an enumeration algorithm for min-cuts, as well as the uniformly directed cuts (see Example~\ref{ex:udc}) considered in \cite{provan1989exact, provan1996paradigm}.
We then conclude our work by combining these insights with poset slicing \cite{Garg06-tcs,garg2001slicing} to give a heuristic for min-cuts satisfying more general constraint structures.
\subsection{Computing Irreducibles}\label{subsec:irreducibles}

In Figure~\ref{fig:irreducibles}, the irreducible min-cuts are those $(S,T)$ for which there exists $u\in V\setminus\{s,t\}$ where $(S,T)$ is the least min-cut, by refinement, with $u \in S$.
This is true in general, as explained by the following corollary of Birkhoff's theorem.
\begin{lemma}[Folklore]\label{lem:moore}
    Let $\Sigma$ be a finite set, and $\mathcal{F}\subseteq 2^\Sigma$ a ring of sets.
    Then, $X\in \mathcal{F}$ is irreducible if and only if there exists $y \in X$ such that $Z\subseteq X$ implies $y \notin Z$ or $Z\notin \mathcal{F}$.
\end{lemma}

Recall $(\mathcal{C}(s,t),\leq) \cong \left(2^{V\setminus\{s,t\}},\subseteq\right)$ via $(S,T)\mapsto S-s$.
So, the lattice of min-cuts is isomorphic to a sublattice of $\left(2^{V\setminus\{s,t\}},\subseteq\right)$, which is some ring of sets.
Applying the previous lemma and the inverse isomorphism shows that the irreducible min-cuts are exactly the min-cuts $(S,T)$ such that for some vertex $v \in S-s$ there exists no min-cut $(S',T') < (S,T)$ with $v \in S'$.
Because of the correspondence between the satisfying preimage of regular predicates and sublattices, this logic extends to any regular predicate (and conjunction thereof) as well.
In summary, for a regular predicate $B:\mathcal{C}(s,t)\to\{\textsf{T},\textsf{F}\}$, the irreducibles of the sublattice of cuts satisfying $B_\text{MC}\wedge B$ is given exactly by,
\begin{equation*}
    \bigcup_{v \in V\setminus\{s,t\}}\left(\bigsqcap\{(S,T)\in\mathcal{C}(s,t)\mid v \in S \wedge (B_\text{MC}\wedge B)(S,T)\}\right).
\end{equation*}
\vspace{-.8\baselineskip}
\begin{algorithm}[t]
    \caption{Computing Irreducible Min-Cuts}\label{alg:irreducibles}
\DontPrintSemicolon
\SetKwRepeat{Do}{do}{while}
\SetKwInOut{KwInput}{input~}
\SetKwInOut{KwOutput}{output~}
\KwInput{$(s,t)$-network $(V, E)$, capacities $c:E\to\mathbb{Z}_+$, and regular  $B:(\mathcal{C}(s,t),\leq)\to\{\textsf{T},\textsf{F}\}$.}
\KwOutput{The set of irreducible min-cuts satisfying $B$.}

\vspace{\baselineskip}

$\mathcal{I} \gets \varnothing$\;
$f^*\gets \texttt{computeMaxFlow}(V,E,c)$\;
\;
\ForEach{$u \in V\setminus\{s,t\}$ in parallel}
{
    $(S,T) \gets \texttt{leastSatMinCut}(f^*,B\wedge B_{\varnothing\implies u})$\;
    $\mathcal{I} \gets \mathcal{I} + (S,T)$
}
\;
\Return $\mathcal{I}$
\end{algorithm}

Observe, for any $v \in V\setminus\{s,t\}$, the least min-cut $(S,T)$ satisfying $B$ with $v \in S$ is exactly the least cut (with respect to refinement) satisfying the conjunction $B_\text{MC}\wedge B\wedge B_{\varnothing \implies v}$ with an implicational cut predicate.
Therefore, we can build the set of irreducibles by independently applying Algorithm~\ref{alg:constrained-cuts} to $B\wedge B_{\varnothing \implies v}$, for each $v \in V\setminus\{s,t\}$.
Such a procedure is shown in Algorithm~\ref{alg:irreducibles}.
Specifically, we precompute a max-flow and then perform $n$ parallel invocations to Algorithm~\ref{alg:constrained-cuts}.
The resulting min-cuts are collected into a set $\mathcal{I}$ and returned.

To compute forbidden vertices in $B\wedge B_{\varnothing\Rightarrow u}$, one uses separate search procedures on $B$ and $B_{\varnothing\Rightarrow u}$.
As the latter is an implicational cut predicate, it has (strong) $O(1)$-advancement with $O(n)$ work.
Therefore, if $B$ has $T_B$-advancement with $W_B$ work, by the analysis in Section~\ref{subsec:llp-constrained} the procedure uses in $O(T_\text{MF} + n T_B)$ time and $O(W_\text{MF} + n^2(m + n + W_{B}))$ work.

If $B$ is a $k$-transition predicate with strong $T_B$-advancement, then the situation improves.
Specifically, from Lemma~\ref{lem:k-trans-conj} it follows that $B\wedge B_{\varnothing\Rightarrow u}$ is a $(k+2)$-transition predicate with strong $T_B(n,m)$-advancement.
And so, by Theorem~\ref{thm:k-trans} the procedure terminates in $O(T_\text{MF} + kT_B)$ time with $O(W_\text{MF} + kn(m + n + W_B))$ work.
Interestingly, if we are only concerned with irreducible min-cuts, effectively making $B$ equal to a predicate which is always true, then the complexity collapses further.
This takes $O(T_\text{MF})$ time with $O(W_\text{MF} + n(m+n))$ work.
Hence, all irreducible min-cuts can be computed in parallel time equal to that of computing a max-flow.

\subsection{Enumerating Min-Cuts Satisfying Regular Predicates}\label{subsec:enumeration}

Now we enumerate the min-cuts satisfying a regular predicate $B$.
The idea will be to precompute the irreducibles via Algorithm~\ref{alg:irreducibles} to then traverse the ideals of the irreducibles using a recursive depth-first search.
Each listed min-cut is a join over an ideal, and so the correctness of this approach is a result of Birkhoff's theorem \cite{birkhoff1937rings}.
A strength of our approach is its generality, as using the lattice structure of the min-cuts allows one to apply our method to enumerate min-cuts satisfying \emph{any} regular predicates.
This allows one to enumerate uniformly directed cuts \cite{provan1989exact,provan1996paradigm} without rederiving any new results, for example.
Before continuing, another method for enumerating the ideals of a poset is presented in \cite{steiner1986algorithm}.
Though their method uses less operations ($O(n)$ vs. $O(n^2)$), our method is simpler and easily parallelized (whereas theirs relies on sequential backtracking procedures).
Other enumeration algorithms for distributive lattices are given in Ch. 14 of \cite{Gar:2015:bk} under the assumption that the lattice is encoded by a \emph{realizer} (see \cite{Trotter92}).

\begin{algorithm}[t]
    \caption{Enumerating Min-Cuts}\label{alg:enumeration}
\DontPrintSemicolon
\SetKwRepeat{Do}{do}{while}
\SetKwInOut{KwInput}{input~}
\SetKwInOut{KwOutput}{output~}
\SetKwFunction{flistNext}{listNext}
\SetKwFunction{firr}{irreducibleMinCuts}
\SetKwFunction{fleast}{leastCuts}
\SetKwProg{fn}{func}{:}{}
\KwInput{$(s,t)$-network $(V, E)$, capacities $c:E\to\mathbb{Z}_+$, and regular $B:(\mathcal{C}(s,t),\leq)\to\{\textsf{T},\textsf{F}\}$.}
\KwOutput{All min-cuts satisfying $B$.}

\vspace{\baselineskip}

$f^*\gets \texttt{computeMaxFlow}(V,E,c)$\;
$(S_\text{bot}, T_\text{bot})\gets \texttt{leastSatMinCut}(f^*,B)$\;
\texttt{listNext}$\left(\{(S_\textup{bot}, T_\textup{bot})\}, \texttt{irreducibles}\left(f^*,B\right)\right)$\;
\;
\fn{\textup{\texttt{listNext}}$\left(\mathcal{K}, \mathcal{I}\right)$}
{
    \textbf{list} $\bigsqcup \mathcal{K}$\;
    \;
    \mbox{$\mathcal{X}\gets\{(S,T)\in\mathcal{I}\mid (\nexists (S',T')\in\mathcal{I})\; (S',T')<(S,T)\}$}\par\nobreak
    \ForEach{$(S,T) \in \mathcal{X}$ in parallel}
    {
        $\mathcal{U}_{(S,T)} \gets \{(S',T'\in\mathcal{I}\mid (S',T') \geq (S,T)\}$\;
    }
    \;
    \ForEach{$\mathcal{L} \subseteq \mathcal{X}$}
    {
        $\mathcal{J} \gets \mathcal{I}\setminus \bigcup_{(S,T) \in \mathcal{X}\setminus\mathcal{L}}\mathcal{U}_{(S,T)}$\;

        \texttt{listNext}$\left(\mathcal{K}\cup\mathcal{L}, \mathcal{J}\right)$\;
    }
}
\end{algorithm}

Algorithm~\ref{alg:enumeration} begins by using a max-flow solution to precompute the irreducibles and least element of the sublattice of min-cuts satisfying $B$.
Then, in each round of recursive call to the function \texttt{listNext}, the min-cut defined by the join $\bigsqcup \mathcal{K}$ is first listed.
The other parameter $\mathcal{I}$ can be understood as the operating set of irreducibles allowed for use in future recursive calls.
Next, the least cuts from $\mathcal{I}$ are computed to form the set $\mathcal{X}$ in line 8.
The least cuts are used to incrementally advance the ideal used to list the cuts.
This is, for each cut $(S,T)\in\mathcal{X}$ we compute the principal filter of $(S,T)$ for storage as $\mathcal{U}_{(S,T)}$ in lines 9-10.
Then, for each subset $\mathcal{L}$ of $\mathcal{X}$ we recurse on $\mathcal{K}\cup \mathcal{L}$, so the next listed cut is that given by the join of $\mathcal{K}\cup\mathcal{L}$, and the set of irreducibles $\mathcal{J}$ formed by subtracting the principal filters of irreducibles in $\mathcal{X}\setminus \mathcal{L}$ from $\mathcal{I}$ in lines 12-14.
The latter mechanism is our means of ensuring that the procedure lists each cut no more than once (needed for correctness).
This is more clear in the proof of Theorem~\ref{thm:enum-correct}.

We examine the complexity per listed element (take note that the number of listed cuts will generally be exponential in $n$ and $m$).
In the literature this is referred to as the \emph{delay}, see \cite{johnson1988generating} for discussion.
First, observe that Lemma~\ref{lem:moore} ensures there are $O(n)$ irreducibles.
So, the join in line 6 is concretely computed as a union over $O(n)$ sets (see Section~\ref{sec:preliminaries}), and thus takes $O(\log n)$ time and $O(n^2)$ work using a reduce algorithm.
Lines 9-10 are computed in $O(1)$ time with $O(n^2)$ work.
Finally, line 13 again requires a union over $O(n)$ sets, so again takes $O(\log n)$ time and $O(n^2)$ work.
So, in total, Algorithm~\ref{alg:enumeration} takes $O(\log n)$ time and $O(n^2)$ work per listed element.
Also, note that the depth of recursion is bounded by the number of irreducibles, which is $O(n)$ by Lemma~\ref{lem:moore}.

Now we argue correctness.
We first show the invariant that $\mathcal{K}$ is an ideal in every invocation of \texttt{listNext}.
Following this, we prove each needed min-cut is listed exactly once.

\begin{lemma}\label{lem:ideal}
    Let $(V,E)$ be a $(s,t)$-flow network, and $B:(\mathcal{C}(s,t),\leq)\to\{\textsf{T},\textsf{F}\}$ a regular predicate.
    Then, each invocation of $\textup{\texttt{listNext}}(\mathcal{K},\mathcal{I})$ is such that $\mathcal{K}$ is an ideal of poset of irreducibles of the sublattice of min-cuts satisfying $B$.
\end{lemma}
\begin{proof}
    Observe that the first invocation of \texttt{listNext} is such that $\mathcal{K}$ is the singleton given by the least min-cut satisfying $B$ (should one exist), while successive recursions always augment $\mathcal{K}$ by minimal irreducibles outside of $\mathcal{K}$.
    This strategy ensures that any irreducible lying between two cuts (with respect to refinement) in $\mathcal{K}$ is also contained in $\mathcal{K}$, thereby making $\mathcal{K}$ an ideal on the poset of irreducibles.
\end{proof}

\begin{theorem}\label{thm:enum-correct}
    Let $(V,E)$ be a $(s,t)$-flow network, and $B:(\mathcal{C}(s,t),\leq)\to\{\textsf{T},\textsf{F}\}$ a regular predicate.
    Then, Algorithm~\ref{alg:enumeration} lists each min-cut satisfying $B$ exactly once.
\end{theorem}
\begin{proof}
    First, by way of contradiction, suppose a min-cut $(S,T)$ is listed twice.
    By Lemma~\ref{lem:ideal} every listed min-cut is the join of irreducibles forming an ideal (see line 7 of Algorithm~\ref{alg:enumeration}).
    Hence, by Birkhoff's theorem, the ideal corresponding to $(S,T)$ must have been passed into \texttt{listNext} (at least) two times.
    Let this ideal be $\mathcal{K}$, so that $(S,T) = \bigsqcup \mathcal{K}$, and consider lowest ancestor in the recursion tree common to the two invocations listing $(S,T)$.
    Refer to this ancestor via $\texttt{listNext}(\mathcal{K}_0,\mathcal{I}_0)$, and let $\mathcal{Z}$ equal the least irreducibles in the difference $\mathcal{K}\setminus\mathcal{K}_0$ (which is non-empty by $\mathcal{K}\neq\mathcal{K}_0$).
    Because $\mathcal{K}$ and $\mathcal{K}_0$ are both ideals, this construction makes $\mathcal{Z}\subseteq \mathcal{X}_0$, where $\mathcal{X}_0$ equals the set $\mathcal{X}$ computed in the invocation $\texttt{listNext}(\mathcal{K}_0,\mathcal{I}_0)$.
    Now, select any $\mathcal{L}_0\subseteq \mathcal{X}_0$, and let $\mathcal{J}_0$ correspond to the $\mathcal{J}$ computed in $\texttt{listNext}(\mathcal{K}_0,\mathcal{I}_0)$ using $\mathcal{L}_0$.
    If $\mathcal{Z}\not\subseteq\mathcal{L}_0$, there is an irreducible in $\mathcal{Z}$ which is not contained in $\mathcal{J}_0$.
    By noting $\mathcal{Z} \cap \mathcal{K}\neq\varnothing$, we see by Lemma~\ref{lem:ideal} it is impossible for any descendant of $\texttt{listNext}(\mathcal{K}_0\cup\mathcal{L}_0,\mathcal{J}_0)$ in the recursion tree to list $(S,T)$.
    Moreover, if $\mathcal{Z}\subsetneq \mathcal{L}_0 $ then $\mathcal{K}_0\cup\mathcal{L}_0$ contains an irreducible not in $\mathcal{Z}$.
    Notice, $\mathcal{X}_0\setminus\mathcal{Z}$ has empty intersection with $\mathcal{K}\setminus\mathcal{K}_0$ by construction, and so it is likewise impossible for any descendant of $\texttt{listNext}(\mathcal{K}_0\cup\mathcal{L}_0,\mathcal{J}_0)$ in the recursion tree to list a join over $\mathcal{K}$, as this parameter increases over successive recursive calls.
    Therefore, $\mathcal{L}_0 = \mathcal{Z}$ is the only possible choice for which the $\texttt{listNext}(\mathcal{K}_0\cup\mathcal{L}_0,\mathcal{J}_0)$ has a descendant in the recursion tree listing a join over $\mathcal{K}$.
    This implies that the invocation $\texttt{listNext}(\mathcal{K}_0\cup\mathcal{L}_0,\mathcal{J}_0)$ is a lower common ancestor than $\texttt{listNext}(\mathcal{K}_0,\mathcal{I}_0)$ to the two invocations listing $(S,T)$.
    However, this contradicts our assumption that $\texttt{listNext}(\mathcal{K}_0,\mathcal{I}_0)$ is the lowest common ancestor.

    Finally, we show it is not possible for the algorithm to miss listing some min-cut satisfying $B$.
    Select such a cut $(S,T)$ and let $\mathcal{K}$ be the unique ideal in the poset of irreducibles of min-cuts satisfying $B$ such that $(S,T) = \bigsqcup \mathcal{K}$.
    We apply induction over $|\mathcal{K}|$.
    Suppose $|\mathcal{K}| = 1$, then it follows that $(S,T)$ is the least min-cut satisfying the predicate.
    This is listed in the first invocation of \texttt{listNext}.
    Now, for the inductive step select any greatest cut $(S', T')$ in $\mathcal{K}$.
    It follows by hypothesis that $\bigsqcup \left(\mathcal{K} - (S',T')\right)$ is listed.
    Examine the invocation of $\texttt{listNext}(\mathcal{K} - (S',T'),\mathcal{I})$ in particular.
    Since $\mathcal{K}$ is an ideal, every irreducible refined by $(S', T')$ is contained in $\mathcal{K}$.
    This implies that $(S', T') \in \mathcal{I}$, as inspecting Algorithm~\ref{alg:enumeration} shows that an irreducible is not in $\mathcal{I}$ only if it is greater than some cut in $\mathcal{K}-(S',T')$.
    Observe then that $(S',T')$ is a least element in $\mathcal{I}\setminus(\mathcal{K}-(S',T'))$.
    Since $\{(S',T')\} \subseteq \mathcal{L}$ will hold after computing $\mathcal{L}$, it is clear that $(S,T) = \bigsqcup \mathcal{K} $ will be listed in an immediate child invocation in the recursion tree.
\end{proof}

\subsection{Solving General Constraints via Poset Slicing}\label{subsec:general}
\begin{algorithm}[t]
    \caption{Min-Cuts with General Predicates}\label{alg:slicing}
\DontPrintSemicolon
\SetKwRepeat{Do}{do}{while}
\SetKwInOut{KwInput}{input~}
\SetKwInOut{KwOutput}{output~}
\KwInput{$B_\text{reg}, B:(\mathcal{C}(s,t),\leq)\to\{\textsf{T},\textsf{F}\}$ with $B_\text{reg}$ regular, and $(s,t)$-network $(V, E)$.}
\KwOutput{A min-cut satisfying $B_\text{reg}\wedge B$, if one exists.}

\vspace{\baselineskip}

\ForEach{min-cut $(S,T)$ satisfying $B_\text{reg}$}
{
    \lIf{$B(S,T)$}
    {
        \Return $(S,T)$
    }
}
\;
\Return $\bot$
\end{algorithm}
Finally, we propose an algorithm to find a min-cut satisfying a more general set of constraints.
Specifically, let $B_\text{reg}$ be any regular predicate and $B$ any kind of predicate, possibly not lattice-linear.
We give an algorithm using poset slicing \cite{garg2001slicing,Garg06-tcs}, a technique for operating over a structured subset of a lattice.
Using Algorithm~\ref{alg:enumeration} we enumerate only the min-cuts satisfying $B_\text{reg}$, and evaluate for satisfaction of the general $B$.
Note, this algorithm is not guaranteed to be efficient as there may be exponentially many min-cuts satisfying $B_\text{reg}$.
Yet, one can select the regular predicate $B_\text{reg}$ as a heuristic to limit the search space size.
One could also select many regular predicates, and apply this method repeatedly for an increased level of flexibility.
This idea is given in Algorithm~\ref{alg:slicing}.
   
For each enumerated cut we evaluate the predicate $B$.
Let the time and work complexities of doing so be $T_B$ and $W_B$.
Then, the complexity is driven by Algorithms~\ref{alg:irreducibles} and~\ref{alg:enumeration}.
Specifically, if there are $R$ min-cuts satisfying $B_\text{reg}$ and one can compute a forbidden vertices in $T_\text{reg}(n,m)$ time with $W_\text{reg}$ work, then Algorithm~\ref{alg:slicing} terminates in $O(T_\text{MF}+nT_\text{reg} + (R\log n)\cdot T_B)$ time using $O(W_\text{MF} + n^2(m+n+W_\text{reg}) + (Rn^2)\cdot W_B)$ work.

\section{Conclusion}

We apply the LLP method to give a unified algorithmic framework for min-cuts under additional constraints.
By viewing min-cuts as a regular predicate, we obtain simple parallel procedures that advance any cut toward feasibility.
Constraints defined by lattice-linear predicates capture structural restrictions and several network constructions used in prior work \cite{provan1989exact,provan1996paradigm,sigal1980stochastic,avramidis1996integrated}.
This flexibility yields immediate applications to computing irreducibles and enumeration, and our $k$-transition and strong-advancement analyses sharpen the resulting complexity bounds.
Extending these ideas beyond lattice-linear predicates, e.g. to the related \emph{equilevel predicates} \cite{DBLP:conf/icdcn/GargS24}, is an exciting direction for future work.


\bibliographystyle{IEEEtran}
\bibliography{fmaster}

\begin{thebibliography}{10}
\providecommand{\url}[1]{#1}
\csname url@samestyle\endcsname
\providecommand{\newblock}{\relax}
\providecommand{\bibinfo}[2]{#2}
\providecommand{\BIBentrySTDinterwordspacing}{\spaceskip=0pt\relax}
\providecommand{\BIBentryALTinterwordstretchfactor}{4}
\providecommand{\BIBentryALTinterwordspacing}{\spaceskip=\fontdimen2\font plus
\BIBentryALTinterwordstretchfactor\fontdimen3\font minus \fontdimen4\font\relax}
\providecommand{\BIBforeignlanguage}[2]{{%
\expandafter\ifx\csname l@#1\endcsname\relax
\typeout{** WARNING: IEEEtran.bst: No hyphenation pattern has been}%
\typeout{** loaded for the language `#1'. Using the pattern for}%
\typeout{** the default language instead.}%
\else
\language=\csname l@#1\endcsname
\fi
#2}}
\providecommand{\BIBdecl}{\relax}
\BIBdecl

\bibitem{picard1980structure}
J.-C. Picard and M.~Queyranne, ``On the structure of all minimum cuts in a network and applications,'' \emph{Mathematical Programming Study}, vol.~13, pp. 8--16, 1980.

\bibitem{queyranne1998minimizing}
M.~Queyranne, ``Minimizing symmetric submodular functions,'' \emph{Mathematical programming}, vol.~82, no.~1, pp. 3--12, 1998.

\bibitem{DBLP:conf/spaa/Garg20}
V.~K. Garg, ``Predicate detection to solve combinatorial optimization problems,'' in \emph{{SPAA} '20: 32nd {ACM} Symposium on Parallelism in Algorithms and Architectures, Virtual Event, USA, July 15-17, 2020}, C.~Scheideler and M.~Spear, Eds.\hskip 1em plus 0.5em minus 0.4em\relax {ACM}, 2020, pp. 235--245.

\bibitem{birkhoff1937rings}
G.~Birkhoff, ``Rings of sets,'' 1937.

\bibitem{garg2001slicing}
V.~K. Garg and N.~Mittal, ``On slicing a distributed computation,'' in \emph{Distributed Computing Systems, 2001. 21st International Conference on.}\hskip 1em plus 0.5em minus 0.4em\relax IEEE, 2001, pp. 322--329.

\bibitem{Garg06-tcs}
V.~K. Garg, ``Algorithmic combinatorics based on slicing posets,'' \emph{Theoretical Computer Science}, vol. 359, no. 1-3, pp. 200--213, 2006.

\bibitem{schrijver2002history}
A.~Schrijver, ``On the history of the transportation and maximum flow problems,'' \emph{Mathematical programming}, vol.~91, no.~3, pp. 437--445, 2002.

\bibitem{yazdi2010max}
S.~S.~T. Yazdi and S.~A. Savari, ``A max-flow/min-cut algorithm for a class of wireless networks,'' in \emph{Proceedings of the Twenty-First Annual ACM-SIAM Symposium on Discrete Algorithms}.\hskip 1em plus 0.5em minus 0.4em\relax SIAM, 2010, pp. 1209--1226.

\bibitem{jensen2022review}
P.~M. Jensen, N.~Jeppesen, A.~B. Dahl, and V.~A. Dahl, ``Review of serial and parallel min-cut/max-flow algorithms for computer vision,'' \emph{IEEE Transactions on Pattern Analysis and Machine Intelligence}, vol.~45, no.~2, pp. 2310--2329, 2022.

\bibitem{ford1956maximal}
L.~R. Ford, Jr. and D.~R. Fulkerson, ``Maximal flow through a network,'' \emph{Canadian Journal of Mathematics}, vol.~8, pp. 399--404, 1956.

\bibitem{edmonds1972theoretical}
J.~Edmonds and R.~M. Karp, ``Theoretical improvements in algorithmic efficiency for network flow problems,'' \emph{Journal of the ACM}, vol.~19, no.~2, pp. 248--264, 1972.

\bibitem{dinitz2006dinitz}
Y.~Dinitz, ``Dinitz’algorithm: The original version and even’s version,'' in \emph{Theoretical Computer Science: Essays in Memory of Shimon Even}.\hskip 1em plus 0.5em minus 0.4em\relax Springer, 2006, pp. 218--240.

\bibitem{goldberg1988new}
A.~V. Goldberg and R.~E. Tarjan, ``A new approach to the maximum-flow problem,'' \emph{Journal of the ACM (JACM)}, vol.~35, no.~4, pp. 921--940, 1988.

\bibitem{orlin2013max}
J.~B. Orlin, ``Max flows in o(nm) time, or better,'' in \emph{Proceedings of the Forty-Fifth Annual ACM Symposium on Theory of Computing (STOC '13)}, 2013, pp. 765--774.

\bibitem{chen2025maximum}
L.~Chen, R.~Kyng, Y.~Liu, R.~Peng, M.~Probst~Gutenberg, and S.~Sachdeva, ``Maximum flow and minimum-cost flow in almost-linear time,'' \emph{Journal of the ACM}, vol.~72, no.~3, pp. 1--103, 2025.

\bibitem{provan1989exact}
J.~S. Provan and V.~G. Kulkarni, ``Exact cuts in networks,'' \emph{Networks}, vol.~19, no.~3, pp. 281--289, 1989.

\bibitem{provan1996paradigm}
J.~S. Provan and D.~R. Shier, ``A paradigm for listing (s,t)-mincuts in graphs,'' \emph{Algorithmica}, vol.~15, no.~4, pp. 351--372, 1996.

\bibitem{sigal1980stochastic}
C.~E. Sigal, A.~A.~B. Pritsker, and J.~J. Solberg, ``The stochastic shortest route problem,'' \emph{Operations research}, vol.~28, no.~5, pp. 1122--1129, 1980.

\bibitem{avramidis1996integrated}
A.~N. Avramidis and J.~R. Wilson, ``Integrated variance reduction strategies for simulation,'' \emph{Operations Research}, vol.~44, no.~2, pp. 327--346, 1996.

\bibitem{de2023finding}
M.~de~Berg, A.~L{\'o}pez~Mart{\'\i}nez, and F.~Spieksma, ``Finding diverse minimum st cuts,'' in \emph{34th International Symposium on Algorithms and Computation (ISAAC 2023)}.\hskip 1em plus 0.5em minus 0.4em\relax Schloss Dagstuhl--Leibniz-Zentrum f{\"u}r Informatik, 2023, pp. 24--1.

\bibitem{de2025finding}
M.~de~Berg, A.~L. Mart{\'\i}nez, and F.~Spieksma, ``Finding diverse solutions in combinatorial problems with a distributive lattice structure,'' \emph{arXiv preprint arXiv:2504.02369}, 2025.

\bibitem{iwamasa2025generalframeworkfindingdiverse}
Y.~Iwamasa, T.~Matsuda, S.~Morihira, and H.~Sumita, ``A general framework for finding diverse solutions via network flow and its applications,'' \emph{arXiv preprint arXiv:2504.17633}, 2025.

\bibitem{ball1983calculating}
M.~O. Ball and J.~S. Provan, ``Calculating bounds on reachability and connectedness in stochastic networks,'' \emph{Networks}, vol.~13, no.~2, pp. 253--278, 1983.

\bibitem{gardner1985algorithm}
M.~Gardner, ``Algorithm to aid in the design of large scale networks.'' \emph{Large Scale Syst.}, vol.~8, no.~2, pp. 147--156, 1985.

\bibitem{henzinger2020finding}
M.~Henzinger, A.~Noe, C.~Schulz, and D.~Strash, ``Finding all global minimum cuts in practice,'' in \emph{28th Annual European Symposium on Algorithms (ESA 2020)}.\hskip 1em plus 0.5em minus 0.4em\relax Schloss Dagstuhl--Leibniz-Zentrum f{\"u}r Informatik, 2020, pp. 59--1.

\bibitem{karger2016enumerating}
D.~R. Karger, ``Enumerating parametric global minimum cuts by random interleaving,'' in \emph{Proceedings of the forty-eighth annual ACM symposium on Theory of Computing}, 2016, pp. 542--555.

\bibitem{beideman2023approximate}
C.~Beideman, K.~Chandrasekaran, and W.~Wang, ``Approximate minimum cuts and their enumeration,'' in \emph{Symposium on Simplicity in Algorithms (SOSA)}.\hskip 1em plus 0.5em minus 0.4em\relax SIAM, 2023, pp. 36--41.

\bibitem{chase1998detection}
C.~M. Chase and V.~K. Garg, ``Detection of global predicates: Techniques and their limitations,'' \emph{Distributed Computing}, vol.~11, no.~4, pp. 191--201, 1998.

\bibitem{Garg:ICDCN22}
V.~K. Garg, ``A lattice linear predicate parallel algorithm for the dynamic programming problems,'' in \emph{Proc. of the Int'l Conf. on Distributed Computing and Networking (ICDCN)}.\hskip 1em plus 0.5em minus 0.4em\relax Delhi, India: Springer-Verlag, 2022.

\bibitem{AlvGar22}
D.~R. Alves and V.~K. Garg, ``Parallel minimum spanning tree algorithms via lattice linear predicate detection,'' in \emph{Proc. Parallel and Distributed Combinatorics and Optimization (PDCO), June 2022}, Lyon, France, 2022.

\bibitem{garg2021lattice}
V.~K. Garg, ``A lattice linear predicate parallel algorithm for the housing market problem,'' in \emph{International Symposium on Stabilizing, Safety, and Security of Distributed Systems}.\hskip 1em plus 0.5em minus 0.4em\relax Springer, 2021, pp. 108--122.

\bibitem{gupta2025tolerance}
A.~T. Gupta and S.~S. Kulkarni, ``Tolerance to asynchrony in algorithms for multiplication and modulo,'' \emph{Theoretical Computer Science}, vol. 1024, p. 114914, 2025.

\bibitem{gupta2021extending}
------, ``Extending lattice linearity for self-stabilizing algorithms,'' in \emph{International Symposium on Stabilizing, Safety, and Security of Distributed Systems}.\hskip 1em plus 0.5em minus 0.4em\relax Springer, 2021, pp. 365--379.

\bibitem{gupta2024tolerance}
------, ``Tolerance to asynchrony of an algorithm for gathering myopic robots on an infinite triangular grid,'' in \emph{2024 19th European Dependable Computing Conference (EDCC)}.\hskip 1em plus 0.5em minus 0.4em\relax IEEE, 2024, pp. 61--68.

\bibitem{gupta2024eventually}
------, ``Eventually lattice-linear algorithms,'' \emph{Journal of Parallel and Distributed Computing}, vol. 185, p. 104802, 2024.

\bibitem{DBLP:conf/icdcn/GargS24}
V.~K. Garg and R.~P. Streit, ``Parallel algorithms for equilevel predicates,'' in \emph{Proceedings of the 25th International Conference on Distributed Computing and Networking, {ICDCN} 2024, Chennai, India, January 4-7, 2024}.\hskip 1em plus 0.5em minus 0.4em\relax {ACM}, 2024, pp. 104--113.

\bibitem{davey}
B.~A. Davey and H.~A. Priestley, \emph{Introduction to Lattices and Order}.\hskip 1em plus 0.5em minus 0.4em\relax Cambridge, UK: Cambridge University Press, 1990.

\bibitem{Gar:2015:bk}
V.~K. Garg, \emph{Lattice Theory with Computer Science Applications}.\hskip 1em plus 0.5em minus 0.4em\relax New York, NY: Wiley, 2015.

\bibitem{kashyap2005intractability}
S.~Kashyap and V.~K. Garg, ``Intractability results in predicate detection,'' \emph{Information processing letters}, vol.~94, no.~6, pp. 277--282, 2005.

\bibitem{serna1991tight}
M.~Serna and P.~Spirakis, ``Tight rnc approximations to max flow,'' in \emph{Annual Symposium on Theoretical Aspects of Computer Science}.\hskip 1em plus 0.5em minus 0.4em\relax Springer, 1991, pp. 118--126.

\bibitem{agarwal2024parallel}
A.~Agarwal, S.~Khanna, H.~Li, P.~Patil, C.~Wang, N.~White, and P.~Zhong, ``Parallel approximate maximum flows in near-linear work and polylogarithmic depth,'' in \emph{Proceedings of the 2024 Annual ACM-SIAM Symposium on Discrete Algorithms (SODA)}.\hskip 1em plus 0.5em minus 0.4em\relax SIAM, 2024, pp. 3997--4061.

\bibitem{korte2008combinatorial}
B.~Korte and J.~Vygen, \emph{Combinatorial optimization: theory and algorithms}.\hskip 1em plus 0.5em minus 0.4em\relax Springer, 2008.

\bibitem{ChanLamp:Snap}
K.~M. Chandy and L.~Lamport, ``Distributed snapshots: Determining global states of distributed systems,'' vol.~3, no.~1, pp. 63--75, Feb. 1985.

\bibitem{steiner1986algorithm}
G.~Steiner, ``An algorithm to generate the ideals of a partial order,'' \emph{Operations Research Letters}, vol.~5, no.~6, pp. 317--320, 1986.

\bibitem{Trotter92}
W.~Trotter, \emph{Combinatorics and Partially Ordered Sets: Dimension Theory}.\hskip 1em plus 0.5em minus 0.4em\relax The Johns Hopkins University Press, 1992.

\bibitem{johnson1988generating}
D.~S. Johnson, M.~Yannakakis, and C.~H. Papadimitriou, ``On generating all maximal independent sets,'' \emph{Information Processing Letters}, vol.~27, no.~3, pp. 119--123, 1988.

\end{thebibliography}

\end{document}